\documentclass{journalA4}  

\usepackage{graphicx,enumerate}

\usepackage{amssymb,amsmath,xcolor}
\usepackage{algorithm}
\usepackage[noend]{algorithmic}

\graphicspath{{images/}}

\newcommand{\etal}{\emph{et al.}}
\newcommand{\eps}{\varepsilon}

\newcommand{\R}{\mathbb{R}}

\newcommand{\f}{Fr\'echet }
\newcommand{\MinimumQuery}{\text{\sc minimumQuery}}
\newcommand{\acro}[1]{\widetilde{#1}}
\newcommand{\opt}[1]{#1^*}
\newcommand{\acropt}[1]{\opt{\acro{#1}}}
\newcommand{\acrotrc}[1]{\overline{#1}}

\newtheorem{invariant}{Invariant}
\newtheorem{corollary}{Corollary}

\begin{document}

\title{Computing the \f Distance with a Retractable Leash}

\author{
Kevin Buchin\thanks{Department of Mathematics and Computer Science,
TU Eindhoven, The Netherlands. {\tt k.a.buchin@tue.nl},
{\tt r.v.leusden@student.tue.nl}, {\tt w.meulemans@tue.nl}.}
\and
Maike Buchin\thanks{Fakult\"at f\"ur Mathematik, Ruhr
Universit\"at Bochum, Germany. {\tt Maike.Buchin@ruhr-uni-bochum.de}.}
\and
Rolf van Leusden\footnotemark[1]
\and
Wouter Meulemans\thanks{giCentre, City University London, United Kingdom.
{\tt wouter.meulemans@city.ac.uk}.  Supported by the
Netherlands Organisation for Scientific Research (NWO) under
project no.~639.022.707.}
\and
Wolfgang Mulzer\thanks{Institut f\"ur Informatik, Freie Universit\"at
Berlin, Germany. {\tt mulzer@inf.fu-berlin.de}. Supported in part
by DFG project MU/3501/1.}
}

\maketitle

\begin{abstract}
All known algorithms for the \f distance 
between curves proceed in two steps: first, 
they construct an efficient oracle for the 
decision version; second, they use this 
oracle to find the optimum from a finite 
set of critical values. We present a novel 
approach that avoids the detour through 
the decision version. This gives the first 
quadratic time algorithm for the \f distance 
between polygonal curves in $\R^d$ under 
polyhedral distance functions (e.g., $L_1$ 
and $L_\infty$). We also get a 
$(1+\eps)$-approximation of the \f 
distance under the Euclidean metric, in 
quadratic time for any fixed $\eps > 0$. For 
the exact Euclidean case, our framework 
currently yields an algorithm with running 
time $O(n^2 \log^2 n)$. However, we 
conjecture that it may eventually lead 
to a faster exact algorithm.
\end{abstract}

\section{Introduction}
Measuring the similarity of curves is a 
classic problem in computational geometry.
For example, it is used for map-matching 
tracking data~\cite{BrakatsoulasPSW05,WenkSP06} 
and moving objects analysis~\cite{BuchinBG10,BuchinBGLL11}.
In these applications, it is important 
to take the continuity of the curves into 
account. Therefore, the \emph{\f distance} 
and its variants are popular metrics to 
quantify (dis)similarity. The \f distance 
between two curves is obtained by taking 
a homeomorphism between the curves that 
minimizes the maximum pairwise distance.
It is commonly explained through the 
\emph{leash}-metaphor: a man walks on 
one curve, his dog walks on the other 
curve.  Man and dog are connected by 
a leash. Both can vary their speeds, but 
they may not walk backwards. The \f distance 
is the length of the shortest leash so that 
man and dog can walk from the beginning 
to the end of the respective curves.

\paragraph{Related work.}
The algorithmic study of the \f distance 
was initiated by Alt and Godau~\cite{AltGo95}.
They gave an algorithm to solve the decision 
version for polygonal curves in $O(n^2)$ 
time, and then used parametric search to 
find the optimum in $O(n^2 \log n)$ time, 
for two polygonal curves of complexity $n$.
The method by Alt and Godau is very general
and also applies to polyhedral distance functions.
To avoid the need for parametric search, 
several randomized algorithms have been 
proposed that are based on the decision 
algorithm combined with random sampling of 
critical values, one running in $O(n^2 \log^2 n)$ 
time~\cite{CookWenk10}, the other in 
$O(n^2 \log n)$ time~\cite{HarPeledRa11}.
Recently, Buchin~\etal~\cite{bbmm-fswd-12} 
showed how to solve the decision version in 
subquadratic time, resulting in a randomized 
algorithm for computing the \f distance in 
$O(n^2 \log^{1/2} n \log\log^{3/2} n)$ time.

In terms of the leash-metaphor, these algorithms 
simply give several leashes to the man and his 
dog to try if a walk is possible. By a clever 
choice of leash-lengths, one then finds 
the \f distance efficiently. Since no 
substantially subquadratic algorithm for 
the problem is known, several faster 
approximation algorithms have been 
proposed (e.g. \cite{AltKW01,DriemelHW10}).
However, these require various assumptions 
of the input curves; previous to our work, there was 
no approximation algorithm that for the 
general case runs faster than known exact 
algorithms. Recently, Bringmann~\cite{bringmann} 
showed that, unless the Strong Exponential 
Time Hypothesis (SETH) fails, no 
general-case $O(n^{2-\alpha})$ algorithm 
can exist to approximate the \f distance 
within a factor of $1.001$, for any $\alpha > 0$. 
The lower bound on the approximation factor was later improved 
to $1.399$, even for the one-dimensional 
discrete case~\cite{bringmannMu}.  Subsequent to
our work, Bringmann and Mulzer showed
that a very simple greedy algorithm
yields an approximation factor of $2^{O(n)}$
in linear time~\cite{bringmannMu}. 
This leaves 
us with a gap between the known algorithms 
and lower bounds for computing and 
approximating the \f distance.

\paragraph{Contribution.}
We present a novel framework for computing 
the \f distance, one that does not rely on 
the decision problem. Instead, we give the 
man a ``retractable leash'' that can be 
lengthened or shortened as required. To this 
end, we consider monotone paths on the 
\emph{distance terrain}, a generalization 
of the \emph{free space diagram} typically 
used for the decision problem. Similar 
concepts have been studied before, but 
without the monotonicity requirement (e.g., 
path planning with height restrictions on 
terrains~\cite{deberg1997} or the 
\emph{weak} \f distance~\cite{AltGo95}).

We present the core ideas for our approach 
in Section~\ref{sec:framework}. The framework 
provides a choice of the distance function 
$\delta$ that is used to measure the distance 
between points on the curves.  However, it 
requires an implementation of a certain data 
structure that depends on $\delta$.  We apply 
our framework
to polyhedral distances (Section~\ref{sec:polyhedral}), 
to show that under such metrics, the \f distance 
is computable in quadratic time. To the best of
our knowledge, there is no previous method
for this case that is faster than the classic 
Alt-Godau algorithm 
with running  time $O(mn \log n)$~\cite{AltGo95}.
Our polyhedral implementation can be used to obtain 
a $(1+\eps)$-approximation for the Euclidean 
case (Section~\ref{sec:approx}). This leads 
to an $O(mn (d + \log \frac{1}{\eps}))$-time 
algorithm, giving the first approximation 
algorithm that runs faster than known exact 
algorithms for the general case.
Moreover, as shown by Bringmann~\cite{bringmann}, 
our result is tight up to subpolynomial factors, assuming 
SETH. 
Finally, we apply 
our framework to the Euclidean distance 
(Section~\ref{sec:euclidean}), to show that 
using this approach, we can compute the \f 
distance in $O(mn (d + \log^2 m + \log^2 n))$ 
time for two $d$-dimensional curves of 
complexity $m$ and $n$.
We conclude with two open problems 
in Section~\ref{sec:conclusion}.

\section{Framework}
\label{sec:framework}

\subsection{Preliminaries}
\label{ssec:prelim}

\paragraph{Curves and distances.}
Consider a curve $P$ in a $d$-dimensional 
space.  We denote the vertices of $P$ by 
$p_0, \dots, p_m$; its complexity (number 
of edges) is $m$. We treat a curve as a 
piecewise-linear function $P \colon [0,m] 
\rightarrow \R^d$. That is, 
$P(i+\lambda) = (1 - \lambda) p_i + \lambda p_{i+1}$ 
holds for any integer 
$i \in \{0, \dots, m-1\}$ and $\lambda \in [0,1]$.
Similarly, we are given a curve 
$Q \colon [0,n] \rightarrow \R^d$ with 
complexity $n$; its vertices are denoted 
by $q_0, \dots, q_n$.

Let $\Psi$ be the set of all orientation-preserving 
homeomorphisms, i.e., continuous and 
nondecreasing functions 
$\psi \colon [0,m] \rightarrow [0,n]$ with 
$\psi(0) = 0$ and $\psi(m) = n$.
Then the \emph{\f distance} is defined as
\[
d_\text{F}(P,Q) = 
\inf_{\psi \in \Psi} \max_{t \in [0,m]}\big\{ \delta\big(P(t), Q(\psi(t))
\big) 
\Big\}.
\]
Here, $\delta$ may be any distance 
function on $\R^d$. Here, we shall 
consider polyhedral distance functions 
(Section~\ref{sec:polyhedral}) and the more
typical case of the Euclidean 
distance function 
(Section~\ref{sec:euclidean}). For our 
framework, we require that $\delta$ is 
\emph{convex}. That is, the locus of all 
points with distance at most one to the 
origin forms a convex set in $\R^d$.

\begin{figure}[b]
\centering
\includegraphics[scale=0.9]{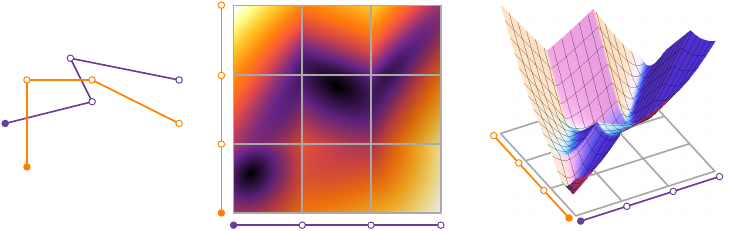}
\caption{Illustration of a distance terrain 
with the Euclidean distance in $\R^2$. Left: 
two curves. Middle: cells as seen from above. 
Dark colors indicate low ``height''. Right: 
perspective view.}
\label{fig:terrain}
\end{figure}

\paragraph{Distance terrain.}
Consider the joint parameter space 
$R = [0,m] \times [0,n]$ of $P$ and $Q$.
A pair $(s,t) \in R$ corresponds to the 
points $P(s)$ and $Q(t)$, and the distance 
function $\delta$ assigns a distance 
$\delta(P(s),Q(t))$ to $(s,t)$. We interpret 
this distance as the ``height'' at point 
$(s,t) \in R$. This gives a \emph{distance 
terrain} $T$, i.e., $T : R \rightarrow \R$ 
with $T(s,t) = \delta(P(s),Q(t))$. We 
partition $T$ into $mn$ \emph{cells} based 
on the vertices of $P$ and $Q$. For 
integers $i \in \{0, \dots, m-1\}$ and 
$j \in \{0, \dots, n-1\}$, the cell 
$C_{i,j}$ is defined as the subset 
$[i,i+1] \times [j,j+1]$ of the parameter 
space $R$. The cells form a regular grid, 
where $i$ represents the column and $j$ the 
row of a cell. The \emph{sides} of $C_{i,j}$ 
are the four line segments 
$[i,i+1] \times \{j\}$, $[i,i+1] \times \{j+1\}$, $\{i\} \times [j,j+1]$, 
and $\{i+1\} \times [j,j+1]$; the \emph{boundary} 
of $C_{i,j}$ is the union of its sides. An 
example of two curves and their distance 
terrain is given in Figure~\ref{fig:terrain}.

A path $\pi : [0,1] \rightarrow R$ is 
\emph{bimonotone} if it is both $x$- and 
$y$-monotone, i.e., every horizontal and 
vertical line intersects $\pi$ in at most 
one connected component. For $(s,t) \in R$, 
we let $\Pi(s,t)$ be the set of all 
bimonotone continuous paths from the origin 
to $(s,t)$. The \emph{acrophobia function} 
$\widetilde{T} : R \rightarrow \R$ is 
defined as
\[
  \acro{T}(s,t) = \inf_{\pi \in \Pi(s,t)} 
                  \max_{\lambda \in [0,1]} T(\pi(\lambda)).
\]
Intuitively, $\acro{T}(s,t)$ represents 
the lowest height that an acrophobic (and 
somewhat neurotic) climber needs to master 
in order to reach $(s,t)$ from the origin 
on a bimonotone path through the distance 
terrain $T$. A bimonotone path from $(0,0)$ 
to $(m,n)$ corresponds to a homeomorphism: we 
have $d_\text{F}(P,Q) = \acro{T}(m,n)$.

Let $x \in R$ and $\pi \in \Pi(x)$ be a 
bimonotone path from the origin to $x$. Let $\eps \geq 0$.
We call $\pi$ an \emph{$\eps$-witness} for 
$x$ if
\[
  \max_{\lambda \in [0,1]} T(\pi(\lambda)) \leq \eps.
\]
For $\eps = \acro{T}(x)$, we call $\pi$ simply 
a \emph{witness}: $\pi$ is then an optimal path 
for the acrophobic climber.

\paragraph{Algorithm strategy.}
Due to the convexity of the distance function, 
we need to consider only the boundaries of 
cells of the distance terrain. It seems natural 
to propagate through the terrain for any point 
on a cell side the minimal ``height'' (leash 
length) $\eps$ required to reach that point.
However, this may entail an amortized linear 
number of changes when moving from one cell 
to the next, giving a cubic-time lower bound 
for such an approach. We therefore do not 
maintain these functions explicitly. Instead, 
we maintain sufficient information to compute 
the lowest $\eps$ for a side. A single pass 
over the terrain then finds the minimum $\eps$ 
for reaching the other end, giving the \f 
distance. 

More specifically, we show that as
we move through a row $j$ of the distance terrain
from left to right, the witnesses for the minimum
values of the acrophobia function on the 
vertical boundaries exhibit a
certain \emph{monotonicity property}: if a witness
for the $i$-th vertical boundary enters row 
$j$ in column $a$, then there is a witness
for the $(i+1)$-th vertical boundary that
enters row $j$ in column $a$ or to the
right of column $a$. Thus, if we know that
the ``rightmost'' witness for the $i$-th
vertical boundary enters row $j$ in column $a$,
it suffices to consider only witnesses that
enter in columns $a, a+1, \dots, i+1$.
Furthermore, we can narrow down the set of 
candidate columns further by observing that
it is enough to restrict our attention to 
those columns for which
the minimum value of the acrophobia function
on the bottom boundary is smaller than for all bottom 
boundaries to the right of it, up to $i+1$ 
(otherwise, we could find an equally good witness further
to the right). Now, all we need is an efficient
way to decide whether for a given candidate column,
there actually exists an optimum witness for
the $(i+1)$-th vertical boundary that enters
row $j$ through this column. For this, we
describe \emph{witness envelopes}, a data structure
that allows us to characterize an optimum witness
that enters row $j$ in a given column. Furthermore,
we show that these witness envelopes can be maintained
efficiently, assuming that an appropriate data structure
for dynamic upper envelopes is available. Putting
everything together, and proceeding analogously for the columns
of the distance terrain,
we obtain a new algorithm for the \f distance.

\subsection{Analysis of the distance terrain}
\label{ssec:analysis}

The \f distance corresponds to the acrophobia 
function $\acro{T}$ on the distance terrain.
To compute $\acro{T}(m,n)$, we show that it 
suffices to consider the cell boundaries. For 
this, we generalize the fact that cells of 
the free space diagram are convex~\cite{AltGo95} 
to the distance terrain for convex distance 
functions.

\begin{lemma}\label{lem:convexthreshold}
Let $\eps \geq 0$, and suppose that $\delta$ 
is a convex distance function. For every cell 
$C$, the set of all points $(s,t) \in C$ with
$T(s,t) \leq \eps$ is convex.
\end{lemma}

\begin{proof}
The cell $C$ represents the parameter space 
of two line segments in $\R^d$. Let $\ell_P(s)$ 
and $\ell_Q(t)$ be the parameterized lines 
spanned by these line segments. Both $\ell_P$ 
and $\ell_Q$ are affine maps. Consider the 
map $f \colon \R^2 \rightarrow \R^d$ defined 
by $f(s,t) = \ell_P(s) - \ell_Q(t)$. Being 
a linear combination of affine maps, $f$ is 
affine. Set 
$D_\eps = \{ z \in \R^d \mid \delta(0,z) \leq \eps \}$.
Since $\delta$ is convex, $D_\eps$ is 
convex. Let $E = f^{-1}(D_\eps)$. Since the 
affine preimage of a convex set is convex, 
$E$ is convex. Thus, $C \cap E$, the subset 
$(s,t) \in C$ with $T(s,t) \leq \eps$, is 
convex, as it is the intersection of two 
convex sets.
\end{proof}

Lemma~\ref{lem:convexthreshold} has two 
important consequences. First, it shows that 
it is indeed enough to focus on cell 
boundaries. Second, it tells us that the 
distance terrain along each side is 
\emph{unimodal}, that is, it has a single 
local minimum.

\begin{corollary}\label{col:witnessEdge}
Let $C$ be a cell of the distance terrain, 
and $x_1$ and $x_2$ two points on different 
sides of $C$. For any $y$ on the line 
segment $x_1x_2$, we have 
$T(y) \leq \max\{T(x_1), T(x_2)\}$.
\end{corollary}

\begin{corollary}\label{col:convexboundary}
Let $C$ be a cell of the distance terrain.
The restriction of $T$ to any side of $C$ 
is unimodal.
\end{corollary}

We denote by $L_{i,j}$ and $B_{i,j}$ the left 
and bottom side of the cell $C_{i,j}$ (and, 
by slight abuse of notation, also the 
restriction of $T$ to the side). The right 
and top side are given by $L_{i+1,j}$ and 
$B_{i,j+1}$.\footnote{Note that there need 
not be an actual cell $C_{i+1,j}$ or 
$C_{i,j+1}$.} With $\acro{L}_{i,j}$ and 
$\acro{B}_{i,j}$ we denote the acrophobia 
function along the corresponding side. All 
these restricted functions depend on a 
single parameter $\alpha \in [0,1]$ in the 
natural way, i.e., $L_{i,j}(\alpha) = 
T(i, j + \alpha)$, $B_{i,j}(\alpha) = T(i + \alpha, j)$, etc.
Assuming that the distance function $\delta$ 
is symmetric, computing values for rows 
and columns of $T$ is symmetric as well.
Hence, we present only how to compute with 
rows. If $\delta$ is asymmetric, our methods 
still work, but some extra care needs to be 
taken when computing distances. In the 
following, we fix a row $j$, and we write 
$C_i$ as a shorthand for $C_{i,j}$, $L_i$ 
for $L_{i,j}$, etc.

Consider a vertical side $L_{i}$. We write 
$\acropt{L}_{i}$ for the minimum of the 
acrophobia function $\acro{L}_{i}$ along 
$L_{i}$, and similarly for horizontal sides.  
Our goal is to compute $\acropt{L}_{i}$ and 
$\acropt{B}_{i}$ for all cell boundaries.
We say that an $\eps$-witness $\pi$ 
\emph{passes through} a side $B_{i}$ if 
there is a $\lambda \in [0,1]$ with 
$\pi(\lambda) \in B_{i}$.

\begin{lemma}\label{lem:moveRighterWitness}
Let $\eps > 0$, and $x$ a point on $L_{i}$.
Let $\pi$ be an $\eps$-witness for $x$ 
that passes through $B_{a}$, for some 
$a \in \{0, \dots,i-1\}$. Suppose there is 
a column $b \in \{a+1, \dots, i-1\}$ with 
$\acropt{B}_{b} \leq \eps$. Then there 
exists an $\eps$-witness for $x$ that 
passes through $B_{b}$.
\end{lemma}

\begin{proof}
Let $y$ be a point on $B_{b}$ that achieves 
$\acropt{B}_{b}$, and $\pi_y$ a witness for 
$y$. Since $\pi$ is bimonotone and passes 
through $B_{a}$, it must also pass through 
$L_{b+1}$. Let $z$ be the (lowest) 
intersection point of $\pi$ and $L_{b+1}$, 
and $\pi_z$ the subpath of $\pi$ from $z$ 
to $x$. Let $\pi'$ be the path obtained by 
concatenating $\pi_y$, the line segment 
$yz$, and $\pi_z$. By our assumption on 
$\eps$ and by Corollary~\ref{col:witnessEdge}, 
path $\pi'$ is an $\eps$-witness for $x$ 
that passes through $B_{b}$; see 
Figure~\ref{fig:rightmost}.
\end{proof}

\begin{figure}[ht]
\centering
\includegraphics{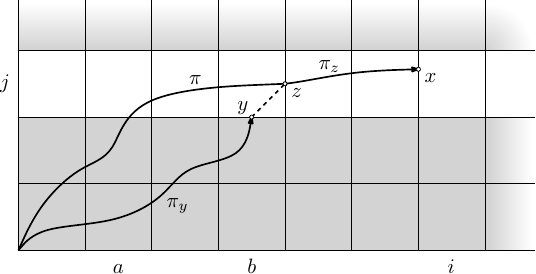}
\caption{Suppose $x \in L_{i}$ has an 
$\eps$-witness that passes through $B_{a}$, 
and $\acropt{B}_{b} \leq \eps$ for some 
$a < b < i$. Then, $x$ has an $\eps$-witness 
that passes through $B_{b}$.}
\label{fig:rightmost}
\end{figure}

Lemma~\ref{lem:moveRighterWitness} implies 
that any point $x \in L_{i}$ has a 
\emph{rightmost} witness $\pi$ with the 
property that if $\pi$ passes through the 
bottom side $B_{a}$, for some $a < i$, then 
the acrophobia function on all later bottom 
sides is strictly greater than the acrophobia 
optimum at $x$.

\begin{corollary}\label{col:rightmost}
Let $x$ be a point on $L_{i}$. There is a
witness $\pi$ for $x$ with the following 
property: if $\pi$ passes through the bottom 
side $B_{a}$, then $\acropt{B}_{b} > \acro{T}(x)$, 
for all $b \in \{ a+1, \dots, i-1 \}$.
\end{corollary}

Next, we argue that there is a witness for 
$\acropt{L}_{i+1}$ that enters row $j$ at 
or after the bottom side used by the 
witness for $\acropt{L}_{i}$. That is, 
the rightmost witnesses behave 
``monotonically'' in the terrain.

\begin{lemma}\label{lem:lowestOptCP}
Let $\pi$ be a witness for $\acropt{L}_{i}$ 
that passes through $B_{a}$, for some 
$a \in \{0, \dots, i-1\}$. Then 
$\acropt{L}_{i+1}$ has a witness that 
passes through $B_{b}$, for some 
$b \in \{a, \dots, i\}$.
\end{lemma}

\begin{proof}
Choose $b$ maximum so that $\acropt{L}_{i+1}$ 
has a witness $\pi'$ that passes through 
$B_{b}$. If $b \geq a$, we are done, so 
assume $b < a$. Since $\pi'$ must pass 
through $L_{i}$, we get 
$\acropt{L}_{i+1} \geq \acropt{L}_{i} \geq \acropt{B}_{a}$.
Lemma~\ref{lem:moveRighterWitness} now 
gives a witness for $\acropt{L}_{i+1}$ that 
passes through $B_{a}$, despite the choice 
of $b$.
\end{proof}

We now characterize $\acro{L}_{i}$ through 
a \emph{witness envelope}. Fix 
$i \in \{1, \dots, m\}$. Suppose 
$\acropt{L}_{i-1}$ has a witness that 
passes through $B_{a'}$. Fix a second 
column $a \in \{a', \dots, i-1\}$. We 
are interested in the best witness for 
$L_{i}$ that passes through $B_{a}$.
The witness envelope is a function 
$\mathcal{E}_{a,i} \colon [0,1] \rightarrow \R$.
The witnesses must pass through $B_{a}$ 
and $L_{i-1}$ (if $a < i - 1$), and 
they end on $L_{i}$. Hence,
\[
  \mathcal{E}_{a,i}(\lambda) \geq 
  \max \{ \acropt{B}_{a}, \acropt{L}_{i-1}, L_{i}(\lambda) \}.
\]
However, this is not enough to exactly 
characterize the best witnesses for $L_{i}$ 
through $B_{a}$. To this end, we introduce 
\emph{truncated terrain functions} 
$\acrotrc{L}_{b}(\lambda) = \min_{\mu \in [0,\lambda]} L_{b}(\mu)$, 
for $b \in \{a +1, \dots, i-1\}$. Since 
$L_{b}$ is unimodal, $\acrotrc{L}_{b}$ 
represents the decreasing part until 
the minimum, remaining constant afterwards.
Therefore,
\[
\mathcal{E}_{a,i}(\lambda) \geq \acrotrc{L}_{b}(\lambda),
\]
for all $b = a+1, \dots, i-1$. The reason 
for truncating the function is as follows: 
to reach $L_{i}(\lambda)$, we must cross 
all $L_{b}$ below $y$-coordinate $j + \lambda$.
If we pass $L_{b}$ below the position where 
the minimum is attained, the height $L_{b}$ 
may force a higher value for the acrophobia 
function. However, the increasing part of 
$L_{b}$ does not matter, because we could 
just pass $L_{b}$ closer to the minimum.
This intuition is not quite accurate, since
we need to account for the order of the 
increasing parts to ensure bimonotonicity.
However, we prove below that due to the witness 
for $\acropt{L}_{i-1}$ through $B_{a'}$, this 
is not a problem. Thus, the witness envelope 
for the column interval $\{a,\dots, i\}$ in 
row $j$ is the upper envelope of the 
following functions on the interval $[0,1]$:
\renewcommand{\labelenumi}{(\roman{enumi})}
\begin{enumerate}
\item the terrain function $L_{i}(\lambda)$;
\item the constant function $\acropt{B}_{a}$;
\item the constant function $\acropt{L}_{i-1}$, 
  if $a \leq i - 2$; and
\item the truncated terrain functions 
  $\acrotrc{L}_{b}(\lambda)$, for all 
  $b = a+1, \dots, i-1$.
\end{enumerate}
\renewcommand{\labelenumi}{\arabic{enumi}.}

\begin{figure}[t]
\centering
\includegraphics{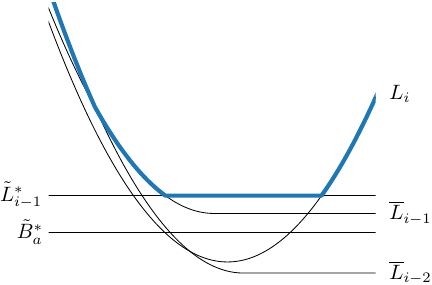}
\caption{A witness envelope for $a = i - 3$. 
It is the upper envelope of two constant 
functions, one (untruncated) terrain function, 
and two truncated terrain functions.}
\label{fig:witnessenvelope}
\end{figure}

See Figure~\ref{fig:witnessenvelope} for an 
example.
We prove with the following lemma that the 
witness envelope exactly characterizes $\acro{L}_{i}$ 
for witnesses that pass through $B_{a}$.

\begin{lemma}\label{lem:witnessenvelope}
Fix a row $j$ and a two columns $a'$, $i$ 
with $a' \leq i-1$. Suppose that 
$\acropt{L}_{i-1}$ has a witness $\pi_{i-1}$ 
that passes through $B_{a'}$. Let $a \in 
\{a', \dots, i-1\}$, $\alpha \in [0,1]$, 
and $\eps > 0$. The point $x = (i, j + \alpha)$ 
has an $\eps$-witness that passes through 
$B_{a}$ if and only if 
$\eps \geq \mathcal{E}_{a,i}(\alpha)$.
\end{lemma}

\begin{proof}
Let $\pi$ be an $\eps$-witness for $x$ that 
passes through $B_{a}$. Then, 
$\eps \geq \acropt{B}_{a}$ and 
$\eps \geq L_{i}(\alpha)$. If $a \leq i-2$, 
then $\pi$ must pass through $L_{i-1}$, so
$\eps \geq \acropt{L}_{i-1}$. Since $\pi$ 
is bimonotone, it has to pass through 
$L_{b}$ for $a < b < i$. Let $y_1 = (a+1, j + \alpha_1), 
y_2 = (a+2, j + \alpha_2), \dots, y_k = (a+k, j + \alpha_k)$ 
be the points of intersection, from left 
to right.  Then, 
$\alpha_1 \leq \alpha_2 \leq \dots \leq \alpha_k \leq \alpha$
and
$\eps \geq T(y_l) = L_{a+l}(\alpha_i) \geq \acrotrc{L}_{a+l}(\alpha)$,
for all $l = 1, \dots, k$.  Hence, 
$\eps \geq \mathcal{E}_{a,i}(\alpha)$.

\begin{figure}[t]
  \centering
  \includegraphics{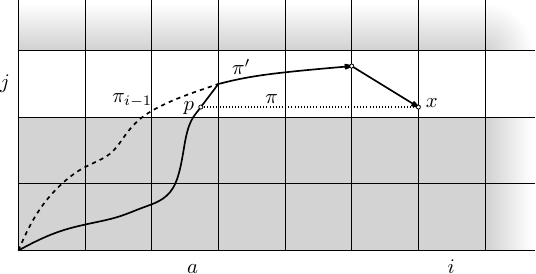}
  \caption{To construct $\pi'$, we combine the
  witness for $\acropt{B}_{a}$, the witness 
  $\pi_{i-1}$, and the segment from $\acropt{L}_{i-1}$ 
  to $x$. By assumption, $\pi_{i-1}$ enters row $j$ at
  or to the left of $B_{a}$. If $\alpha < \alpha'$, 
  then $\pi'$ is not bimonotone, and we shortcut with 
  segment $px$ (dotted) to obtain $\pi$.}
  \label{fig:witnessproof}
\end{figure}

Now suppose that $\eps \geq \mathcal{E}_{a,i}(\alpha)$.
The conclusion is immediate for $a = i-1$.
Otherwise, we have $\eps \geq \acropt{L}_{i-1}$.
Let $\alpha'$ be such that the witness $\pi_{i-1}$ 
for $\acropt{L}_{i-1}$ reaches $L_{i-1}$ at point 
$(i-1, j+\alpha')$. There are two cases. First, 
if $\alpha \geq \alpha'$, we can find an appropriate
$\eps$-witness $\pi'$ for $x$ by following the witness 
for $\acropt{B}_{a}$, passing to $\pi_{i-1}$, 
following $\pi_{i-1}$ to $\acropt{L}_{i-1}$, and 
then taking the line segment to $x$. Second, if 
$\alpha < \alpha'$, we construct a curve $\pi'$ as 
before. However, $\pi'$ is not bimonotone (the last 
line segment goes down). This is fixed as follows: 
let $p$ and $x$ be the two intersection points of 
$\pi'$ with the horizontal line $y = j + \alpha$. 
We shortcut $\pi'$ at the line segment $px$ as 
illustrated in Figure~\ref{fig:witnessproof}.
The resulting curve $\pi$ is bimonotone and passes 
through $B_{a}$. To see that $\pi$ is an $\eps$-witness, 
it suffices to check that along the segment $px$, 
the distance terrain never goes above $\eps$.
For this, we need to consider only the intersections 
of $px$ with the vertical sides. Let $L_{b}$ be such 
a side. The function $L_{b}$ is unimodal; let 
$\opt{\alpha}$ be the value where the minimum of 
$L_{b}$ is obtained. We distinguish two cases to 
argue that $\eps \geq L_{b}(\alpha)$ and to prove 
the lemma:
\begin{enumerate}
  \item $\alpha \leq \opt{\alpha}$: by definition 
  of truncated terrain functions, 
  $L_{b}(\beta) = \acrotrc{L}_{b}(\beta)$, for 
  all $\beta \in [0, \opt{\alpha}]$. Hence, we 
  know that $\eps \geq L_{b}(\alpha)$ holds trivially 
  by our assumption of $\eps \geq \mathcal{E}_{a,i}(\alpha)$ 
  and the fact that $\acrotrc{L}_{b}$ is part of 
  the witness envelope. 
  \item $\alpha \geq \opt{\alpha}$: by construction, 
  the witness $\pi_{i-1}$ passes $L_{b}$ at $\alpha$ 
  or higher. Hence, $\acropt{L}_{i-1} \geq L_{b}(\alpha)$ 
  holds as $L_{b}(\alpha)$ is on the increasing 
  part of $L_{b}$. It follows that 
  $\max\{\acrotrc{L}_{b}(\alpha), \acropt{L}_{i-1}\} \geq L_{b}(\alpha)$.
  Since $\eps \geq 
  \mathcal{E}_{a,i}(\alpha) \geq 
  \max\{\acrotrc{L}_{b}(\alpha), \acropt{L}_{i-1}\}$, 
  we have $\eps \geq L_{b}(\alpha)$, as desired.
\end{enumerate}
Thus, $\pi$ passes through $B_{a}$ and is an 
$\eps$-witness for $x$.
\end{proof}

\subsection{Algorithm}
\label{ssec:algorithm}

We are now ready to present the algorithm.
We walk through the distance terrain, row 
by row, in each row from left to right.
When processing a cell $C_{i,j}$, we compute 
$\acropt{L}_{i+1,j}$ and $\acropt{B}_{i,j+1}$. 
For each row $j$, we maintain a double-ended 
queue (deque) $Q_j$ that stores a sequence 
of column indices. We also store a data 
structure $U_j$ that contains a set of 
(truncated) terrain functions on the vertical 
sides in row $j$. The structure $U_j$ supports insertion, 
deletion, and a minimum-point query that 
returns the lowest point on the upper envelope 
of the terrain functions. In other words, $U_j$ 
implicitly represents a witness envelope, apart 
from the constant functions $\acropt{B}_{a}$ 
and $\acropt{L}_{i-1}$. The implementation of 
$U_j$ depends on the distance function 
$\delta$: in Section~\ref{sec:polyhedral}, we 
describe the data structure for polyhedral
distance functions, and in Section~\ref{sec:euclidean},
we consider the Euclidean case.

\begin{algorithm}[b]
  \caption{$\textsc{FrechetDistance}(P, Q, \delta)$}
  \label{alg:frechetbasic}
  \begin{algorithmic}[1]

  \REQUIRE $P$ and $Q$ are polygonal curves 
    with $m$ and $n$ edges in $\R^d$;\\
    $\delta$ is a convex distance function in $\R^d$
  \ENSURE \f distance $d_\text{F}(P,Q)$ for $\delta$

  \smallskip

  \COMMENT {We show computations only within a row, 
  column computations are analogous}

  \STATE $\acropt{L}_{0,0} \gets \delta(P(0), Q(0))$
  \STATE $\acropt{L}_{0,j} \gets \infty$ for all $j = 1, \dots, n-1$
  \STATE For each row $j$, create empty deque $Q_j$ and upper envelope
    structure $U_j$

  \FOR{$j \gets 0$ \TO $n-1$; $i \gets 0$ \TO $m-1$}

     \STATE Remove all values $x$ from $Q_j$ with 
     $\acropt{B}_{x,j} \geq \acropt{B}_{i,j}$ and 
     append $i$ to $Q_j$ \label{line:clearqueue}
     \IF {$|Q_j| = 1$}
        \STATE Clear $U_j$
     \ENDIF
     \STATE Add $L_{i+1,j}$ to $U_j$ \label{line:insertenvelope}

     \STATE Let $h$ and $h'$ be the first and second element in $Q_j$
     \STATE $(\alpha, \eps_\alpha) \gets 
     U_j.\MinimumQuery()$ \label{line:query1}
     \STATE $\eps_\alpha \gets 
     \max\{\eps_\alpha, \acropt{L}_{i,j}, \acropt{B}_{h,j} \}$
     \WHILE {$|Q_j| \geq 2$ \AND $\acropt{B}_{h',j} \leq \eps_\alpha$}
       \label{line:while}
        \STATE Remove all $L_{x,j}$ from $U_j$ with 
       	$x \leq h'$ \label{line:deleteenvelope}
        \STATE Remove the head $h$ from $Q_j$
        \STATE Let $h$ and $h'$ be the first and second element in $Q_j$
        \STATE $(\alpha, \eps_\alpha) \gets 
	U_j.\MinimumQuery()$ \label{line:query2}
        \STATE $\eps_\alpha \gets 
	\max\{\eps_\alpha, \acropt{L}_{i,j}, \acropt{B}_{h,j} \}$
     \ENDWHILE
     \STATE $\acropt{L}_{i+1,j} \gets \eps_\alpha$
     \STATE Update $L_{i+1,j}$ to $\acrotrc{L}_{i+1,j}$ in $U_j$
     \label{line:trunc}
  \ENDFOR
  \RETURN $\max\{ \delta(P(m), Q(n)), 
    \min \{ \acropt{L}_{m-1,n-1}, \acropt{B}_{m-1,n-1} \} \}$
  \end{algorithmic}
\end{algorithm}

The algorithm is given in Algorithm~\ref{alg:frechetbasic}.
It proceeds as follows: since all witnesses 
start at $(0,0)$, we initialize $C_{0,0}$ to 
use $(0,0)$ as its lowest point and compute 
the distance accordingly. The left- and 
bottommost sides of the distance terrain 
are considered unreachable.

In the body of the for-loop, we compute 
$\acropt{L}_{i+1,j}$ and $\acropt{B}_{i,j+1}$.
Let us describe how to find $\acropt{L}_{i+1,j}$.
First, we remove all indices from the back of 
the $Q_j$ that have an acrophobia optimum on 
the bottom side that is at least $\acropt{B}_{i,j}$,
and we append $i$ to $Q_j$. We also add 
$L_{i+1,j}$ to the upper envelope $U_j$. Let 
$h$ and $h'$ be the first two elements of $Q_j$.
We perform a minimum query on the witness envelope, 
combining the result with two constants 
$\acropt{L}_{i,j}$ and $\acropt{B}_{h,j}$, in order 
to find the smallest $\eps_\alpha$ for which a 
point on $L_{i+1,j}$ has an $\eps_\alpha$-witness 
that passes through $B_{h,j}$. Note that 
$\acropt{L}_{i,j}$ should be included as a 
constant only if $h < i$, i.e., if $|Q_j| \geq 2$; 
for simplicity, we omit this detail in the overview.
If $\eps_\alpha \geq \acropt{B}_{h', j}$, there is 
an $\eps_\alpha$-witness for $L_{i+1,j}$ through 
$B_{h', j}$, so we can repeat the process with 
$h'$ (after updating $U_j$). If $h'$ does not 
exist (i.e., $|Q_j| = 1$) or if 
$\eps_\alpha < \acropt{B}_{h', j}$, we stop and 
declare $\eps_\alpha$ to be optimal. Finally, we 
update $U_j$ to use the truncated terrain function 
$\acrotrc{L}_{i+1,j}$ instead of $L_{i+1,j}$.

We now give the invariant that holds at the 
beginning of each iteration of the for-loop.
The invariant is stated only for a row, analogous 
data structures and invariants apply to the columns.
A point $(\alpha, \beta) \in \R^2$ \emph{dominates} 
a point $(\gamma,\delta) \in \R^2$ if 
$\alpha > \gamma$ and $\beta \leq \delta$.
As before, we from now on fix a row $j$, and we omit 
the index $j$ from all variables.

\begin{invariant}\label{inv:algo}
At the beginning of iteration $i+1$ in row $j$, 
we have computed the optima $\acropt{L}_{1}$, 
$\acropt{L}_{2}$, $\dots$, $\acropt{L}_{i}$.
Let $a$ be the column such that a rightmost 
witness for $\acropt{L}_{i}$ passes through $B_{a}$.
Then $Q$ stores the first coordinates of the points 
in the sequence $(a, \acropt{B}_{a})$, 
$(a+1, \acropt{B}_{a+1})$, $\dots$, $(i-1, \acropt{B}_{i-1})$ 
that are not dominated by any other point in 
the sequence. In addition, $U$ stores the 
(truncated) terrain functions for the vertical 
sides in columns $a+1, \dots, i$.
\end{invariant}

Invariant~\ref{inv:algo} holds initially, so 
we need to prove that it is maintained in 
each iteration of the for-loop. This is done 
in the following lemma.

\begin{lemma}\label{lem:basiccorrect}
Algorithm~\ref{alg:frechetbasic} maintains 
Invariant~\ref{inv:algo}.
\end{lemma}

\begin{proof}
By the invariant, a rightmost witness for 
$\acropt{L}_{i}$ passes through $B_{h_0}$, 
where $h_0$ is the head of $Q$ at the 
beginning of the iteration. Let $h^*$ be 
the column such that a rightmost witness for
$\acropt{L}_{i+1}$ passes through $B_{h^*}$. 
Then $h^*$ is contained in $Q$ after $i$ 
has been added, because by Lemma~\ref{lem:lowestOptCP},
we have $h^* \in \{h_0, \dots, i\}$, and by 
Corollary~\ref{col:rightmost}, there can be 
no column index $a \in \{h^*+1, \dots, i\}$ 
that dominates $(h^*, \acropt{B}_{h^*})$.

Now let $h$ be the head of $Q$ before a minimum 
query on $U$, and $h'$ the second element of $Q$.
By Lemma~\ref{lem:witnessenvelope}, the
minimum query gives the smallest $\eps_\alpha$
for which there is an $\eps_\alpha$-witness for
$L_{i+1}$ that passes through $B_{h}$.
If $h < h^*$, then 
$\eps_\alpha \geq  \acropt{L}_{i+1}$ 
(definition of $\acropt{L}$);
$\acropt{L}_{i+1} \geq \acropt{B}_{h^*}$ 
(there is a witness through $B_{h^*}$);  
and $\acropt{B}_{h^*} \geq \acropt{B}_{h'}$ 
(the dominance relation ensures that the 
$\acropt{B}$-values for the indices in $Q$ 
are increasing). Thus, the while-loop in 
line~\ref{line:while} proceeds to the
next iteration. If $h = h^*$, then by 
Corollary~\ref{col:rightmost}, we have 
$\acropt{B}_{a} > \acropt{B}_{h^*}$ for all
$a \in \{h^* + 1, \dots, i\}$, and the 
while-loop terminates with the correct 
value for $\acropt{L}_{i}$. It is 
straightforward to check that 
Algorithm~\ref{alg:frechetbasic} maintains
the data structures $Q$ and $U$ according 
to the invariant.
\end{proof}

\begin{theorem}\label{thm:basicalg}
Let $\delta$ be a convex distance function 
in $\R^d$. Algorithm~\ref{alg:frechetbasic} 
computes $d_\text{F}(P,Q)$ for $\delta$ in time
$O(mn (T_\text{ue}(m,d,\delta) + T_\text{ue}(n,d,\delta)))$,
where $T_\text{ue}$ represents the 
time to insert into, delete from, and query 
the upper envelope data structure.
\end{theorem}

\begin{proof}
Correctness follows from Lemma~\ref{lem:basiccorrect}.
For the running time, observe that we insert 
each column index only once into $Q$ and 
each terrain function at most twice into $U$
(once untruncated, once truncated). Hence, 
we can remove elements at most once or 
twice. This results in an amortized running 
time of $O(1 + T_\text{ue}(n,d,\delta) + T_\text{ue}(m,d,\delta))$ 
for a single iteration of the for-loop. Since 
there are $O(mn)$ cells, this results in 
the claimed total execution time, assuming 
that $T_\text{ue}$ is $\Omega(1)$.
\end{proof}

\subsection{Avoiding Truncated Functions}

In Algorithm~\ref{alg:frechetbasic}, the envelope 
$U$ uses the (full) unimodal distance function 
only for $L_{i+1}$ and the truncated versions 
for the other cells. Since our algorithm relies
on an efficient data structure to maintain dynamic
upper envelopes of these distance functions, and
since it is easier to design such a data structure
if the set of possible functions to be 
stored is limited, we would like to avoid the
need for truncating the functions. In general,
this seems hard to do, but we show here
that as long as the functions behave 
like pseudolines (i.e., each pair of functions 
intersects at most once, and this intersection is proper), 
we can actually
work with the simpler set of untruncated distance
functions. Since we compare only 
functions in the same row (or column), functions 
in different rows or columns may still intersect 
more than once. Using the full unimodal functions 
potentially allows for a more efficient 
implementation of the envelope structure.

The idea is as follows: since the terrain 
distance functions on the cell boundaries
are unimodal, the initial (from left to right) envelopes
of the truncated distance functions and
the untruncated distance functions are
identical. The two envelopes begin to
differ only when the increasing part
of an untruncated distance function
``cuts off'' a part of the envelope.
We analyse our algorithm
to understand under which circumstances this 
situation can occur. It turns out
that in most cases, the increasing parts
of the distance functions are ``hidden''
by the inclusion of the constant
$\acropt{L}_{i}$ in the witness
envelope, except for one case,
namely when the deletion of a
distance function from the
witness envelope exposes 
an increasing part of a distance
function that did not previously appear
on the envelope. However, we will
see that this case can 
be detected easily, and that 
it can be handled by 
simply removing the increasing distance
function from the upper envelope.
The fact that the distance
functions behave like pseudolines
ensures that the removed function does
not play any role in later queries
to the witness envelope. This idea is formalized and proven below.

We modify Algorithm~\ref{alg:frechetbasic} as 
follows: we omit the update to $U$ in 
line~\ref{line:trunc}, thus $U$ maintains 
untruncated, unimodal functions. To perform 
a minimum-point query, we first run the query 
on the upper envelope of the full unimodal 
functions. Let $(\alpha, \eps_{\alpha})$ 
be the resulting minimum. If $(\alpha, \eps_{\alpha})$  
lies on the intersection of an increasing 
$L_{a}$ and a decreasing $L_{b}$ with $a < b$, 
we remove $L_{a}$ from $U$ and repeat the 
query. Otherwise, we return $\eps_{\alpha}$, 
which is then again combined with the 
constants $\acropt{L}_{i}$ and $\acropt{B}_{h}$ 
as usual.

Below, we prove that this modified algorithm 
is indeed correct. Let $U$ be the envelope 
maintained by the modified algorithm (with 
full functions), and $\overline{U}$ the 
envelope of the original algorithm (with 
truncated functions). We let both 
$\overline{U}$ and $U$ include the constants 
$\acropt{L}_{i}$ and $\acropt{B}_{h}$. The 
envelopes $U$ and $\overline{U}$ are 
unimodal: they consist of a decreasing part, 
(possibly) followed by an increasing part.
Let $D$ and $\overline{D}$ be the decreasing 
parts of $U$ and $\overline{U}$, up to the 
global minimum.

First, we make the following observation. 
With it, we prove that $D$ and $\overline{D}$ 
are identical throughout the algorithm 
(Invariant~\ref{inv:notrunc}).

\begin{lemma}\label{lem:lopt}
Fix a terrain function $L_{a}$. Let $i \geq a$ 
such that $L_{a}$ is contained in $\overline{U}$ 
at the end of iteration $i$.
Then 
$\acropt{L}_{i} \geq \acropt{L}_{a} \geq \min_\lambda L_{a}(\lambda)$.
\end{lemma}

\begin{proof}
By Invariant~\ref{inv:algo}, there is 
a witness for $\acropt{L}_{i}$ through $L_{a}$.
\end{proof}

\begin{invariant}\label{inv:notrunc}
Suppose we run the original and the modified 
algorithm simultaneously. Then, after each 
minimum query, $\overline{D}$ and $D$ are 
identical. Furthermore, any function that 
the modified algorithm deletes during a 
minimum query does not appear on 
$\overline{D}$ in any future iteration.
\end{invariant}

\begin{proof}
Initially, Invariant~\ref{inv:notrunc} trivially 
holds as the upper envelopes are empty. The 
envelopes $U$ and $\overline{U}$ are modified 
when:
\begin{enumerate}[(a)]
\item inserting a full unimodal terrain 
function (line~\ref{line:insertenvelope});
\item truncating a terrain function 
(line~\ref{line:trunc});
\item deleting a terrain function while 
updating the queue $Q$ 
(line~\ref{line:deleteenvelope}).
\end{enumerate}
We now prove that each case indeed maintains 
the invariant.

\textbf{Case (a)}:
The invariant tells us that $\overline{D}$ 
and $D$ are identical before adding a full 
unimodal terrain function, $L_{i+1}$. Hence, 
$L_{i+1}$ affects $\overline{D}$ and $D$ in 
the same manner (either by adding a piece 
or by shortening them) and 
Invariant~\ref{inv:notrunc} is maintained.

\textbf{Case (b)}:
The truncated part of $L_{i+1}$ is the 
increasing part and hence does not belong 
to $\overline{D}$. As the iteration ends, 
$i$ is increased by one, and $\acropt{L}_{i+1}$ 
is now included in the upper envelope 
rather than $\acropt{L}_{i}$. In the 
truncated envelope $\overline{U}$, the 
value of $\acropt{L}_{i+1}$ is determined by 
$\overline{D}$ and the increasing part of 
$L_{i+1}$. Hence, the minimum remains the 
same when truncating $L_{i+1}$, and 
$\overline{D}$ is unchanged. The modified 
algorithm skips the truncation step, so $D$ 
is not changed. Again, Invariant~\ref{inv:notrunc} 
is maintained.

\textbf{Case (c)}:
After deleting a function from $U$ and 
$\overline{U}$, Invariant~\ref{inv:notrunc} 
may get violated. Although the invariant 
guarantees that all functions on $\overline{D}$ 
are stored by the modified algorithm, it may 
happen that $D$ is cut off by the increasing 
part of a function that is truncated in 
$\overline{U}$. In this case, let the minimum 
$p = (p_x, p_y)$ of $D$ be the intersection of 
the increasing part of $L_{a}$ and the decreasing 
part of $L_{b}$ in iteration $i$. There are two 
subcases: (c1) $b < a < i$; or (c2) $a < b < i$.

Case (c1) cannot occur: during iteration $a-1$,
both the decreasing part of $L_{b}$ and the 
increasing part of $L_{a}$ are present in 
$\overline{U}$. Thus, $\acropt{L}_{a} \geq p_y$, 
and $\acropt{L}_{i} \geq p_y$, by 
Lemma~\ref{lem:lopt}. Therefore, $D$ cannot 
be a proper prefix of $\overline{D}$. In case 
(c2), the modified query algorithm deletes 
$L_{a}$ from $U$ and repeats. If we argue that 
$\acrotrc{L_a}$ does not occur on $\overline{D}$ 
in any future iteration, the algorithm 
eventually stops with $D$ and $\overline{D}$ 
identical, and with Invariant~\ref{inv:notrunc} 
maintained. For this, observe that (i) $a < b$ 
and the decreasing part of $L_{a}$ lies below 
$\acrotrc{L}_{b}$; and (ii) by Lemma~\ref{lem:lopt}, 
$\acropt{L}_{i} \geq \min_\lambda L_{a}(\lambda)$ 
for any iteration $i \geq a$ in which $L_{a}$ 
is contained in $U$. Thus, $\acrotrc{L}_{a}$ 
always lies below $\overline{D}$.
\end{proof}

Now that we have established the desired invariant, 
the following theorem can be stated as a direct 
consequence of it.

\begin{theorem} \label{thm:unimodalpseudolines}
Let $j$ be a row of the distance terrain such 
that the distance functions in row $j$ intersect 
pairwise at most once. Then the minima computed 
by the modified algorithm are identical to the 
minima computed by the original algorithm.
\end{theorem}

\section{Polyhedral distance}
\label{sec:polyhedral}

We consider the \f distance with a convex polyhedral 
distance function $\delta$, i.e., the ``unit sphere'' 
of $\delta$ is a convex polytope in $\R^d$ that 
strictly contains the origin. For instance, the $L_1$ 
and $L_\infty$ distance are polyhedral with the 
cross-polytope and the hypercube as respective unit 
spheres. Throughout, we assume that $\delta$ has 
\emph{complexity} $k$, i.e., its polytope (unit 
sphere) has $k$ facets. The polytope of $\delta$ is 
not required to be regular or symmetric, but as 
before, we simplify the presentation by assuming 
symmetry.

Intuitively, the distance $\delta(u,v)$ is the 
smallest scaling factor $s \geq 0$ such that $v$ 
lies on the polytope, centered on $u$ and scaled by 
a factor of $s$. We compute it as follows. Let 
$\mathcal{F}$ denote the facets of the polytope of 
$\delta$. Let $\delta_f(u,v)$ denote the \emph{facet 
distance} for facet $f \in \mathcal{F}$, that is, 
the multiplicative factor by which the hyperplane 
spanned by $f$ needs to be scaled from $u$ to 
contain $v$. We assume that a facet $f$ is defined 
through the point $p_f$ on the hyperplane spanned 
by $f$ that is closest to the origin: the vector 
from the origin to $p_f$ is normal to $f$. The 
distance $\delta_f(u,v)$ is then computed as 
$p_f \cdot (v - u) / \|p_f\|^2$. This distance 
may be negative, but there is always at least one 
facet with non-negative distance. Then 
$\delta(u,v) = \max_{f\in\mathcal{F}} \delta_f(u,v)$, 
the \emph{maximum} over all facet distances. For a 
general polytope, we can compute the facet distance 
in $O(d)$ and the distance between points in $O(k d)$ 
time. However, for specific polytopes, we may do 
better. To make this explicit in our analysis, we 
denote the time to compute the facet distance by 
$T_\text{facet}(\delta)$.

The distance terrain functions $L_{i,j}$ and 
$B_{i,j}$ are piecewise linear for a convex polyhedral 
distance function $\delta$. Each linear part 
corresponds to a facet of $\delta$. Therefore, 
it has at most $k$ parts. Moreover, for a fixed line 
segment (i.e., within the same row or column), each 
facet has a fixed slope: the parts for this facet 
are parallel. Depending on the polytope, the maximum 
number of parts of a single function may be less than 
$k$. We denote this actual maximum number of parts 
by $k'$. Computing the linear parts of a distance 
terrain function $L_{i,j}$ or $B_{i,j}$ requires 
computing which facets may occur. We denote the time 
it takes to compute the $k'$ relevant facets for a 
given boundary by $T_\text{part}(\delta)$.

We give three approaches.
First, we use an upper envelope structure as in 
the Euclidean case, but exploiting that the distance 
functions are now piecewise linear. Second, we use a 
brute-force approach which is more efficient for small 
to moderate dimension $d$ and complexity $k$. Third, 
we combine these methods to deal with the case of 
moderately sized $d$ and $k'$ being much smaller than $k$.

\paragraph{Upper envelope data structure.}
As $L_{i,j}$ and $B_{i,j}$ are piecewise linear, we 
need a data structure that dynamically maintains the 
upper envelope of lines under insertions, deletions, 
and minimal-point queries. Note that the minimal 
point query now requires us to compute the actual 
minimal point on the upper envelope of lines (instead
of parabolas). We apply the same duality transformation 
as in the Euclidean case and maintain a dynamic convex 
hull. That is, every line $\ell : y = a x + b$ on the 
upper envelope dualizes to a point $\ell^* = (a, -b)$.
Any point $p = (a,b)$ dualizes to a line $p^* : y = a x - b$.
If a point $p$ is above a line $\ell$, then the point $\ell^*$ 
is above the line $p^*$. Hence, the upper envelope 
corresponds to the dual lower convex hull. Since the 
minimum of the upper envelope occurs when the slopes 
change from negative to nonnegative, it dualizes to 
the line segment on the convex hull that intersects 
the $y$-axis. The fastest known data structure for 
this problem is due to Chan \cite{chan2012}: for $h$ 
lines, it has an $O(\log^{1+\tau} h)$ query and 
amortized update time, for any $\tau > 0$.

However, in our case, we can do slightly better by 
using the data structure by Brodal and 
Jacob~\cite{brodal2002}. This data structure does 
not support the minimal-point query directly. However, 
we can make it work by observing that we must insert 
and delete up to $k'$ linear functions each time; it 
is acceptable to run multiple queries as well.

\begin{lemma}\label{lem:brodal_jacob}
We can implement an upper envelope data structure 
structure on $h$ piecewise linear functions of complexity 
at most $k'$ with an amortized update time of 
$O(T_\text{part}(\delta) + k' T_\text{facet}(\delta) + k' \log (hk'))$ 
and a minimal-point query time of $O(k' \log (hk'))$.
\end{lemma}

\begin{proof}
First, we consider insertions and deletions.
Every function is piecewise linear with at most 
$k'$ parts, so there are at most $hk'$ lines in 
the data structure. Hence, it takes $O(k' \log hk')$ 
amortized time to insert and delete the parts of a 
single function. To compute the $k'$ relevant lines 
that make up the piecewise linear function, we first 
find the $k'$ relevant facets of $\delta$ in 
$O(T_\text{part}(\delta))$ time. Then we compute 
the parameters of the corresponding lines by computing 
for each relevant facet $f$ the distance between 
$P(i)$ and $Q(j)$ and $Q(j+1)$ with respect to $f$.
This takes $O(T_\text{facet}(\delta))$ time per facet.

For the minimal-point query, we observe that the 
lines with positive slope (that is, dual points 
with positive $x$-coordinate) are truncated at the 
end of each iteration. Hence, at any point during 
the algorithm, the dual lower hull contains at most 
$k'$ points with positive $x$-coordinate.
We maintain only the points with nonpositive $x$-coordinate 
(lines with negative slope) in the data structure.
To find the line segment that intersects the $y$-axis, 
we perform for each current point with positive 
$x$-coordinate a tangent query in the convex hull 
structure. We maintain the tangent with the lowest 
intersection with the $y$-axis: this tangent gives 
the intersection between the $y$-axis and the actual 
lower hull (including the points with positive 
$x$-coordinate). We perform $k'$ queries, each in 
$O(\log hk')$ time; a minimal-point query takes 
$O(k' \log hk')$ time.
\end{proof}

\paragraph{Brute-force approach.}
A very simple data structure can often lead to 
good results. Here, we describe such a data 
structure, exploiting that in a single row, the 
distance function for each facet has a fixed 
slope. Unlike the other approaches, this method 
does not require computing the $k'$ relevant facets 
and thus not depend on $T_\text{part}(\delta)$.

\begin{lemma}\label{lem:brute_force}
After $O((m+n) k (T_\text{facet}(\delta) + \log k))$ 
total preprocessing time,  we can implement 
the upper envelope structure with an amortized 
update and query time of $O(k T_\text{facet}(\delta))$.
\end{lemma}

\begin{proof}
During the preprocessing phase, we sort for each 
segment of $P$ and $Q$ the facets of $\delta$ 
by the corresponding slope on the witness envelope.
This takes $O((m+n) k (T_\text{facet}(\delta) + \log k))$ 
total time using the straightforward algorithm.

Consider the upper envelope data structure $U_j$ 
for a row $j$ (columns are again analogous).
Structure $U_j$ must represent a number of unimodal 
functions, each consisting of a number of linear 
parts. Each linear part corresponds to a certain 
facet of the polytope and has a fixed slope.
For each facet $f \in \{1, \dots, k\}$ (in sorted 
order), structure $U_j$ stores a doubly linked list 
$F_f$ containing lines spanned by these linear parts.
Given the fixed slope, lines in a single list 
$F_f$ do not intersect and are sorted from top to 
bottom. The upper envelope is fully determined 
only by top lines in each list $F_f$.

When processing a cell boundary $L_{i,j}$, we 
update each list $F_l$ in $U_j$: remove all lines 
below the line for $P(i)$ from the back of $F_l$, 
and append the line for $P(i)$. Per facet, it 
takes $O(T_\text{facet}(\delta))$ time to compute 
the $y$-intersection of the line and amortized 
$O(1)$ time for the insertion. We then go through 
the top lines in the $F_l$ in sorted order to 
determine the minimal value on the upper envelope 
in $O(k)$ time.
\end{proof}

\paragraph{A hybrid approach.}
We can combine the methods from 
Lemma~\ref{lem:brodal_jacob} and 
Lemma~\ref{lem:brute_force} into a hybrid approach.

\begin{lemma}\label{lem:hybrid}
After $O((m+n) k)$ total preprocessing time, we 
can implement the upper envelope structure with 
amortized update time 
$O(T_\text{part}(\delta) + k' T_\text{facet}(\delta) + k' \log k)$ 
and minimal-point query time $O(k' \log k)$.
\end{lemma}

\begin{proof}
For each row (or column), we initialize $k$ empty 
lists $F_l$, $l = 1, \dots, k$. This takes $O((m+n)k)$ 
total preprocessing time. The role of the $F_l$ is 
similar to Lemma~\ref{lem:brute_force}, i.e., each 
list $F_l$ corresponds to a facet of the polytope.
However, unlike Lemma~\ref{lem:brute_force}, we do 
not sort the facets. Instead, we maintain the upper 
envelope of the top lines in each $F_l$, using the 
method from Lemma~\ref{lem:brodal_jacob}. At each cell 
boundary, we find the $k'$ relevant parts and compute 
their parameters. The parts are inserted into the 
appropriate lists $F_l$. If a new part appears at 
the top of its list, we update the upper envelope 
structure. Since now this structure stores only 
$k$ lines, this takes amortized time 
$O(T_\text{part}(\delta) + k' T_\text{facet}(\delta) + k' \log k)$.

Minimal-point queries are done as before (see the proof 
of Lemma~\ref{lem:brodal_jacob}). Again, the structure 
contains only $k$ lines: a query takes $O(k' \log k)$ time.
\end{proof}

Plugging Lemmas~\ref{lem:brodal_jacob}, 
\ref{lem:brute_force}, and \ref{lem:hybrid} into 
Theorem~\ref{thm:basicalg} yields the following 
result. The method that works best depends on the 
chosen polytope and on the given complexity and 
dimensions, that is, on the relationship between $n$, 
$k$, $k'$ and $d$.

\begin{theorem}\label{thm:polyhedral}
Let $\delta$ be a convex polyhedral distance 
function of complexity $k$ in $\R^d$. 
Algorithm~\ref{alg:frechetbasic} computes 
the \f distance under $\delta$ in
\[ O\left( \min \left\{
\begin{array}{c}
mn (T_\text{part}(\delta) + k' T_\text{facet}(\delta) + 
  k' \log (m n k')), \\
(m+n) k \log k + mn k T_\text{facet}(\delta), \\
(m+n) k + mn (T_\text{part}(\delta) + 
k' T_\text{facet}(\delta) + k' \log k)
\end{array} \right\} \right)\]
time, where $T_\text{part}(\delta)$ is the time 
needed to find the relevant parts of a distance 
function and $T_\text{facet}(\delta)$ the time 
needed to compute the distance between two points 
for a given facet of $\delta$.
\end{theorem}

\begin{proof}
The first bound follows directly from 
Lemma~\ref{lem:brodal_jacob} and Theorem~\ref{thm:basicalg}.
For the second bound, use Lemma~\ref{lem:brute_force} 
and observe that $(m+n) k T_\text{facet}(\delta)$ 
is asymptotically smaller than $mn k T_\text{facet}(\delta)$.
For the last bound, use  Lemma~\ref{lem:hybrid}.
\end{proof}

For a generic polytope, we have 
$T_\text{facet}(\delta) = O(d)$, so the 
brute-force approach runs in $O(n k \log k + n^2 k d)$ 
time. The other methods can be faster only 
if $k' = o(k)$ and if we have an $o(k d)$-time 
method to compute the relevant facets for a 
distance terrain function. The hybrid method 
improves over the upper-envelope method if 
$k'$ is much smaller than $k$. Note that 
there cannot be more than 
$\min \{ k , n k' \}$ elements in the upper 
envelope for the hybrid method. However, if 
$k > n k'$, the upper-envelope method 
outperforms the hybrid method. Thus, to gain 
an advantage over the brute force method, 
a structured polytope is necessary.

\begin{corollary}{\label{col:polygeneric}}
Let $\delta$ be a convex polyhedral distance 
function of complexity $k$ in $\R^d$. 
Algorithm~\ref{alg:frechetbasic} computes 
the \f distance under $\delta$ in 
$O((m+n) k \log k + mn k d)$ time.
\end{corollary}

Let us now consider $L_\infty$. Its 
polytope is the hypercube; each facet 
is determined by a maximum coordinate. 
We have $k' \leq k = 2d$, and the 
brute-force method outperforms the 
other methods. However, a facet depends 
on only one dimension, so we compute 
the distance for a given facet in 
$T_\text{facet}(L_\infty) = O(1)$ time.

\begin{corollary}{\label{col:linf}}
Algorithm~\ref{alg:frechetbasic} computes 
the \f distance under the $L_\infty$ 
distance in $\R^d$ in $O((m+n) d \log d + m n d)$ 
time.
\end{corollary}

For $L_1$, the cross-polytope, there are 
$k = 2^d$ facets. Structural insights help 
us improve upon the brute-force method. 
The $2^d$ facets of the cross-polytope are 
determined by the signs of the coordinates.
Let $\ell = Q(j)Q(j+1)$ be the line segment
and $p = P(i)$ the point defining the 
terrain distance $L_{i,j}$. At the 
breakpoints between the parts of $L_{i,j}$, 
one of the coordinates of $\ell - p$ 
changes sign. Therefore, there are at 
most $k' = d+1$ parts. We find these parts 
efficiently by computing for each 
coordinate the point on $\ell - p$ where 
the coordinate becomes zero (if any).
Sorting these values gives a representation 
of the relevant facets in $O(d \log d)$ 
time. The actual facets can then by 
computed in $T_\text{part}(L_1) = O(d^2)$ 
time. Computing the facet distance takes 
$T_\text{facet}(L_1) = O(d)$ time, as for 
a general polytope. We conclude that the 
hybrid approach outperforms the 
brute-force approach. Whether the hybrid 
method outperforms the ``pure'' 
upper-envelope method depends on the 
dimension $d$.

\begin{corollary}{\label{col:l1}}
Algorithm~\ref{alg:frechetbasic} computes 
the \f distance under the $L_1$ distance in $\R^d$ in
\[
  O(\min \{ m n (d^2 + d\log (m n)), (m + n) 2^d + m n d^2 \})
\]
time.
\end{corollary}
\begin{proof}
From the arguments above and from 
Theorem~\ref{thm:polyhedral}, we know that 
the hybrid method runs in 
$O((m + n) 2^d + m n d^2)$ time. Similarly, 
the upper-envelope method runs in 
$O(m n (d^2 + d \log (m n d)))$ time.
Simplification of the 
latter gives $O(m n (d^2 + d\log (m n)))$.
\end{proof}

\section{Approximating the Euclidean distance}
\label{sec:approx}

We can use polyhedral distance functions to 
approximate the Euclidean distance. This 
allows us to obtain the following result.

\begin{corollary} \label{col:approxEuclid}
Algorithm~\ref{alg:frechetbasic} computes a 
$(1+\eps)$-approximation of the \f distance 
under the Euclidean distance in $\R^d$ in 
$O(m n (d + \eps^{-1/2}))$ time.
\end{corollary}

\begin{proof}
A line segment $\ell$ and a point $p$ span 
exactly one plane in $\R^d$ (unless they 
are collinear, in which case we pick an 
arbitrary plane). On this plane, the 
Euclidean unit sphere $O$ is a circle; the 
same circle for each plane. We approximate 
$O$ with a $k$-regular inscribed polygon 
$\overline{O}$ in $\R^2$. We need to orient 
this polygon consistently for all points 
$p$, e.g., by having one side parallel to 
$\ell$. Simple geometry shows that for 
$k = O(\eps^{-1/2})$, the polygon 
$\overline{O}$ is a $(1+\eps)$-approximation 
to $O$. The computation is two-dimensional, 
but we must find the appropriate 
transformations, which takes $O(d)$ time per 
boundary. We no longer need to sort the 
facets of the polytope for each edge; the 
order is given by $\overline{O}$. This 
saves a logarithmic factor for the 
initialization of the brute-force method.
This method performs best and, using 
Theorem~\ref{thm:polyhedral}, we get an 
execution time of 
$O(m n (d + \eps^{-1/2}) + (m+n) \eps^{-1/2} \log \eps^{-1/2})$.
However, for $\eps^{-1/2} \geq  \log^2 m + \log^2 n$, 
we simply compute the exact \f distance in 
$O(m n (d + \log^2 m + \log^2 n))$ time by Theorem~\ref{thm:euclalgo}.
\end{proof}

Though this paper focuses on avoiding the 
decision-and-search paradigm, we can do better 
if we are willing to invoke an algorithm for 
the decision version of the \f distance problem.
\begin{corollary}
We can calculate a $(1 + \eps)$-approximation 
of the \f distance under the Euclidean 
distance in $O(m n d + T_\text{dec}(n,d) \log \eps^{-1})$ 
time, where $T_\text{dec}(n,d)$ is the time 
needed to solve the decision problem for the 
\f distance.
\end{corollary}

\begin{proof}
Corollary~\ref{col:approxEuclid} gives a 
$\sqrt{2}$-approximation to the Euclidean 
distance in $O(m n d)$ time. Then, we go 
from a $\sqrt{2}$-approximation to a 
$(1+\eps)$-approximation by binary search, 
using the decision algorithm.
\end{proof}

Solving the decision version takes 
$T_\text{dec}(n,d) = O(m n d)$ time~\cite{AltGo95}.
For $d = 2$ and the right relation between 
$m$ and $n$, one can do slightly 
better~\cite{bbmm-fswd-12,thesisMeulemans-2014}:
on a pointer machine, we may solve the decision 
version in $O(m n (\log \log n)^{3/2} / \sqrt{\log n})$, 
assuming $m \leq n$ and $m = \Omega(\log^3 n)$;
using a word RAM, we may solve it in 
$O(m n (\log \log n)^{2} / \log n)$, 
assuming $m \leq n$ and $m = \Omega(\log^6 n)$.

\section{Euclidean distance}
\label{sec:euclidean}

Let us now consider our framework under the 
Euclidean distance $\delta_\text{E}$. The 
framework applies, because $\delta_\text{E}$ 
is convex (and symmetric). In fact, we use 
the squared Euclidean distance 
$\delta_\text{E}^2 = \delta_\text{E}(x,y)^2$.
Since squaring is a monotone function on 
$\R^+_0$, computing the \f distance for the 
squared Euclidean distance is equivalent 
to the Euclidean case: if 
$\eps = d_\text{F}(P,Q)$ for 
$\delta_\text{E}^2$, then 
$\sqrt{\eps} = d_\text{F}(P,Q)$ for 
$\delta_\text{E}$. We show that the terrain 
functions for $\delta_\text{E}^2$ in each 
row and column behave like pseudolines. We 
consider only the vertical sides; 
horizontal sides are analogous.

\begin{lemma}\label{lem:sqrEucl}
For $\delta = \delta_\text{E}^2$, each 
distance terrain function $L_{i,j}$ is part 
of a parabola. Any two functions $L_{i,j}$ 
and $L_{i',j}$ intersect at most once.
\end{lemma}

\begin{proof}
The function $L_{i,j}$ represents the 
squared Euclidean distance between the 
point $p = P(i)$ and the line segment 
$\ell = Q(j)Q(j+1)$. Let $\ell'$ be 
the line though $\ell$, uniformly 
parameterized by $\lambda \in \R$, 
i.e., 
$\ell'(\lambda) = \lambda(Q(j+1)- Q(j)) + Q(j)$.
Let $\lambda_p$ be the $\lambda$ for 
which $\ell'(\lambda)$ is closest to 
$p$. By the Pythagorean theorem, 
$L_{i,j}(\lambda) = 
\| \ell'(\lambda) - \ell'(\lambda_p) \|^2 + \|\ell'(\lambda_p) - p\|^2 $.
By the parametrization of $\ell'$, we 
have
\[
    \|\ell'(\lambda) - \ell'(\lambda_p)\|^2 =
    \|\ell\|^2 (\lambda - \lambda_p)^2 =
    \|\ell\|^2 \lambda^2 - 2 \|\ell\|^2 \lambda_p \lambda + 
      \|\ell\|^2 \lambda_p^2.
\]
Hence, $L_{i,j}$ is a parabolic function 
in $\lambda$, where the quadratic term 
depends only on $\ell$. For two functions 
in the same row, this term is the same, 
and thus the parabolas intersect at most 
once.
\end{proof}

By Theorem~\ref{thm:unimodalpseudolines} and 
the above lemma, we can use the modified 
algorithm to maintain $U_j$ with the full 
parabolas rather than truncated ones. The 
parabolas of a single row share the same 
quadratic term, so we can treat them as 
lines by subtracting $\|\ell\|^2 \lambda^2$.
In this transformed space, the constant 
functions $\acropt{L}_{i,j}$ and $\acropt{B}_{h,j}$ 
are now downward parabolas. This causes 
no problems, as these are needed only 
\emph{after} computing the minimum on the 
upper envelope: we can add the term 
$\|\ell\|^2 \lambda^2$ back to the answer 
before these constant functions are needed.

However, the minimum of the upper envelope 
of the parabolas does not necessarily 
correspond to the minimum of the lines. Hence, 
we need a ``special'' minimal-point query that 
computes the minimal point on the parabolas, 
using the upper envelope of the lines. The 
advantage of this transformation is that, 
by treating parabolas as lines, we may 
implement $U_j$ with a standard data structure 
for dynamic half-plane intersection or, 
dually, dynamic convex hull. The fastest such 
structure is due to Brodal and 
Jacob~\cite{brodal2002}, but it does not 
explicitly represent the upper envelope.
It is not clear if it can be modified to 
support our special minimal-point query.\footnote{Due
to the complexity of the data structure of Brodal and 
Jacob~\cite{brodal2002}, it seems to be a formidable
task to adapt it to our needs. However, there are
simpler, slightly suboptimal, data structures for dynamic
planar convex hulls that may be more amenable to 
modification~\cite{BrodalJa00,KaplanTaTi01}.
This would immediately lead to a better running time
for our algorithm.
}
Therefore, we use the slightly slower structure 
by Overmars and Van Leeuwen~\cite{overmars1981}, 
giving $O(\log^2 h)$ time insertions and 
deletions, for a structure containing $h$ lines 
(parabolas). Most importantly, we may compute 
the answer to the special minimal-point query 
in $O(\log h)$ time.

\begin{lemma}\label{lem:parabolaUE}
A minimal-point query on the upper envelope 
of $h$ lines can be implemented in $O(\log h)$ 
time.
\end{lemma}
\begin{proof}
The data structure by Overmars and Van Leeuwen 
maintains a \emph{concatenable queue} for the 
upper envelope. A concatenable queue is an 
abstract data type providing the operations 
\emph{insert}, \emph{delete}, \emph{concatenate} 
and \emph{split}. If the queue is implemented 
with a red-black tree, all these operations take 
$O(\log h)$ time. In addition to the tree, we 
maintain a doubly-linked list that stores the 
elements in sorted order, together with 
cross-pointers between the corresponding nodes 
in the tree and in the list. The list and the 
cross-pointers can be updated with constant 
overhead. Furthermore, the list enables us to 
perform predecessor and successor queries in 
$O(1)$ time, provided that a pointer to the 
appropriate node is available.

The order of the points on the convex hull 
corresponds directly to the order of the lines, 
and hence of the parabolas, on their respective 
upper envelopes. We use the red-black tree to 
perform a binary search for a minimal point on 
the upper envelope $\mathcal{U}$ of the parabolas.
We cannot decide how to proceed solely based on 
the parabola $p$ of a single node. However, using 
the predecessor and successor of $p$, we compute 
the local intersection pattern to guide the binary 
search. This is detailed below.

Let $p$ be a parabola on $\mathcal{U}$; let $l$ 
be the predecessor and $r$ the successor of $p$.
Let $p^*$, $l^*$, and $r^*$ denote their respective 
minima. For $z \in \R^2$, let $x(z)$ be the 
$x$-coordinate of $z$. The parabolas $p,l,r$ 
pairwise intersect exactly once. Let $p_l = p \cap l$ 
and $p_r = p \cap r$. As $l$ and $r$ are the 
neighbors of $p$ on $\mathcal{U}$, we have 
$x(p_l) \leq x(p_r)$; the part of $p$ on 
$\mathcal{U}$ is between $p_l$ and $p_r$.
We distinguish three cases (see Figure~\ref{fig:BScases}):
(i) $x(p^*) \leq x(p_l) \leq x(l^*)$; 
(ii) $x(p_l) \geq x(l^*), x(p^*)$; and 
(iii) $x(p_l) \leq x(p^*)$.
We cannot have $x(l^*) \leq x(p_l) \leq x(p^*)$: 
this would imply that $l$ is above $p$ right of 
$p_l$, although $l$ is the predecessor of $p$.

\begin{figure}[t]
  \centering
  \includegraphics{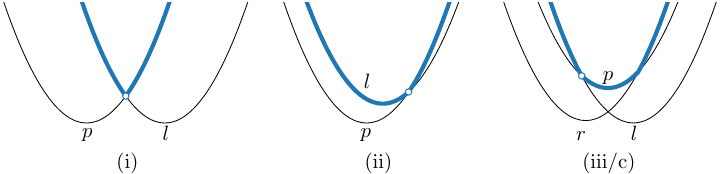}
  \caption{(i) The minimum is $p_l$; 
  (ii) $p_l$ excludes the possibility that the 
  minimum lies on or right of $p$; (iii/c) The 
  minimum cannot be left of $p$. If the analogous 
  case applies to $p_r$, the minimum of $p$ is 
  the minimum of $\mathcal{U}$.}
  \label{fig:BScases}
\end{figure}

In case (i), $l$ is decreasing and $p$ is increasing 
at $p_l$, so $p_l$ is the minimum of $\mathcal{U}$.
In case (ii), $p$ and the part of $\mathcal{U}$ 
right of $p$ do not contain the minimum, as $l$ is 
increasing to the left of $p_l$. Hence, we recurse 
on the left child of $p$.

In case (iii), the part of $\mathcal{U}$ left of 
$p_l$ is higher than $p$. We now consider the 
analogous cases for $p_r$: 
(a) $x(r^*) \leq x(p_r) \leq x(p^*)$; 
(b) $x(p_r) \leq x(p^*), x(r^*)$; and 
(c) $x(p_r) \geq x(p^*)$.
In case (a), $p_r$ is the minimum.
In case (b), we recurse on the right child of $p$.
In case (c), we get $x(p_l) \leq x(p^*) \leq x(p_r)$, 
so $p^*$ is the minimum of $\mathcal{U}$.

As we can access the predecessor and successor 
of a node and determine the intersection pattern 
in constant time, a minimal-point query takes 
$O(\log h)$ time.
\end{proof}

Though not necessary for our algorithm, we 
observe that we may actually obtain the leftmost 
minimal value (after including the constant 
functions) by performing another binary search.
We obtain the following theorem.

\begin{theorem}\label{thm:euclalgo}
Algorithm~\ref{alg:frechetbasic} computes the 
\f distance under the Euclidean distance in 
$\R^d$ in $O(m n (d + \log^2 mn))$ time.
\end{theorem}
\begin{proof}
Lemma~\ref{lem:sqrEucl} implies that we may 
use the modified algorithm 
(Theorem~\ref{thm:unimodalpseudolines}).
For each insertion, we have to compute the 
corresponding parabola, in $O(d)$ time.
The data structure by Overmars and Van 
Leeuwen \cite{overmars1981} allows us to 
implement the dynamic upper envelope of 
$h$ functions with $O(\log^2 h)$-time 
insertions and deletions. The special 
minimal-point query (Lemma~\ref{lem:parabolaUE}) 
takes only $O(\log h)$ time. Hence, 
$T_\text{ue}(h,d,\delta) = O(d + \log^2 h)$ and 
Theorem~\ref{thm:basicalg} implies a total 
execution time of $O(m n (d + \log^2 m + \log^2 n))$.
Since $\log^2 m + \log^2 n = \log^2 (m n) - 2 \log m \log n$, 
the execution time can be simplified to 
$O(m n (d + \log^2 (m n)))$.
\end{proof}

Theorem~\ref{thm:euclalgo} gives a slightly 
slower bound than known results for the 
Euclidean metric. However, we think that 
our framework has potential for a faster 
algorithm (see Section~\ref{sec:conclusion}).

\section{Conclusions and open problems}
\label{sec:conclusion}

We introduced a new method to compute the 
\f distance. It avoids using a decision 
algorithm and its consequence: a search 
on critical values. There is no need for 
parametric search.  For polyhedral 
distance functions we gave an $O(m n)$-time 
algorithm. The implementation of this 
algorithm borders the trivial: the most 
advanced data structure is a doubly linked list.
In addition, it can be used to compute a 
$(1+\eps)$-approximation of the Euclidean
\f distance in 
$O(m n / \sqrt{\eps})$ time or even in 
$O(m n \log \eps^{-1})$ time, if we are 
willing to use a decision algorithm. 
For the exact Euclidean case, we obtain a slightly 
slower running time of
$O\big(mn (\log^2 m + \log^2 n)\big)$. This requires 
dynamic convex hulls and does not really 
improve ease of implementation.
Below, 
we propose two open problems for further 
research. For simplicity, we assume here 
that the two curves have the same 
complexity, that is, $m = n$.

\paragraph{Faster Euclidean distance.}
We think that our current method has room for 
improvement; we conjecture that it is possible 
to extend on these ideas to obtain an $O(n^2)$ 
algorithm for the Euclidean case, at least for 
curves in the plane. Currently we use the full 
power of dynamic upper envelopes, which does not 
seem necessary since all the information about 
the distance terrain functions is available 
in advance.

For points in the plane, we can determine the 
order in which the parabolas occur on the upper 
envelopes, in $O(n^2)$ time for all boundaries.
From the proof of Lemma~\ref{lem:sqrEucl}, we 
know that the order is given by the projection 
of the vertices onto the line. We compute the 
arrangement of the lines dual to the vertices 
of a curve in $O(n^2)$ time. We then determine 
the order of the projected points by traversing 
the zone of a vertical line. This takes $O(n)$ 
time for one row or column. Unfortunately, this 
alone is insufficient to obtain the quadratic time 
bound.

\paragraph{Locally correct \f matchings.}
A \emph{\f matching} is a homeomorphism 
$\psi \in \Psi$ such that it is a witness for 
the \f distance, i.e., 
$\max_{t \in [0,n]} \delta(P(t), Q(\psi(t)) = d_\text{F}(P,Q)$.
A \f matching that induces a \f matching for 
any two matched subcurves is called a \emph{locally 
correct} \f matching~\cite{BuchinBMS12}. It enforces 
a relatively ``tight'' matching, even if the 
distances are locally much smaller than the \f 
distance of the complete curves. The algorithm by 
Buchin~\etal~\cite{BuchinBMS12} incurs a linear 
overhead on the algorithm of Alt and Godau~\cite{AltGo95}, 
resulting in $O(n^3 \log n)$ running time.

The \emph{discrete} \f distance is characterized 
by measuring distances only at vertices. A locally 
correct discrete \f matching can be computed without 
asymptotic overhead by extending the dynamic program 
to compute the discrete \f distance \cite{BuchinBMS12}.
Our algorithm for the (continuous) \f distance is 
much closer in nature to this dynamic program than 
to the decision-and-search paradigm of previously 
known methods. Therefore, we conjecture that our 
framework is able to avoid the linear overhead in 
computing a locally correct \f matching. However, 
the information we currently propagate is insufficient:
a large distance early on may ``obscure'' the rest 
of the computations, making it hard to decide 
which path would be locally correct.

\bibliographystyle{abbrv}
\bibliography{frechet}

\end{document}